%% file: Main.tex
\documentclass[runningheads]{llncs}

\usepackage{amssymb, amsmath, mathrsfs}
\usepackage{synttree, bussproofs,scrextend, stmaryrd, rotating, xcolor, algorithm2e, upgreek}
\usepackage{esvect,pgf, tikz, color}
\usetikzlibrary{arrows, automata}
\usepackage[all]{xy}
\usepackage{changepage,enumerate,proof,upgreek,relsize,trfsigns,
comment}
\input{notation.tex}

\begin{document}

\title{Nested Sequents for Intuitionistic Modal Logics via Structural Refinement\thanks{Work supported by the European Research Council (ERC) Consolidator Grant 771779 %\textit{A Grand Unified Theory of Decidability in Logic-Based Knowledge Representation} 
 (DeciGUT).}}

\author{Tim S. Lyon\orcidID{0000-0003-3214-0828}}

\institute{Computational Logic Group, Institute of Artificial Intelligence, Technische Universit\"at Dresden, Germany  \\ \email{timothy\_stephen.lyon@tu-dresden.de}
}

%======================================================
%Here are our suggested authorrunning and titlerunning:
%======================================================
%\titlerunning{}
%\authorrunning{}

\maketitle

\begin{abstract}
We employ a recently developed methodology---called \emph{structural refinement}---to extract nested sequent systems for a sizable class of intuitionistic modal logics from their respective labelled sequent systems. %In a wide sense, this method can be seen as a means of transforming the semantics of a logic into a nested sequent system for the logic by taking a detour through an associated labelled calculus. In a narrower sense, 
 This method can be seen as a means by which labelled sequent systems can be transformed into nested sequent systems through the introduction of propagation rules and the elimination of structural rules, followed by a notational translation. The nested systems we obtain incorporate propagation rules that are parameterized with formal grammars, and which encode certain frame conditions expressible as first-order Horn formulae that correspond to a subclass of the Scott-Lemmon axioms. We show that our nested systems are sound, cut-free complete, and admit hp-admissibility of typical structural rules.
\keywords{Bi-relational model · Intuitionistic modal logic · Labelled sequent · Nested sequent · Proof theory · Propagation rule · Refinement}
\end{abstract}

\section{Introduction}\label{sec:introduction}
\input{introduction.tex}

%Introduce Propositional Logics, Semantics, Axiomatizations
\section{Logical Preliminaries}\label{sec:log-prelims}

\input{body-1.tex}

%Formal Grammar Stuff used in propagation rules
\section{Grammar Theoretic Preliminaries}\label{sec:grammar-theory}

\input{body-2.tex}

%Labelled Sequents a la Simpson
\section{Labelled Sequent Systems}\label{sec:labelled-systems}

\input{body-3.tex}

%Refinement
\section{Structural Refinement}\label{sec:refinement}

\input{body-4.tex}

%Translation to Nested Sequents
\section{Nested Sequent Systems}\label{sec:nested-calculi}

\input{body-5.tex}

%Conclude
\section{Conclusion}\label{sec:conclusion}

\input{conclusion.tex}

\bibliographystyle{splncs04}
\bibliography{bibliography}

\end{document}

%% file: notation.tex
\DeclareSymbolFont{extraup}{U}{zavm}{m}{n}
\DeclareMathSymbol{\vardiamond}{\mathalpha}{extraup}{87}

\newcommand{\ntr}{\mathfrak{N}}
\newcommand{\ltr}{\mathfrak{L}}
\newcommand{\seqcomp}{\otimes}

\newcommand{\ifandonlyif}{\textit{iff} }
\newcommand{\etc}{$\ldots$ }
\newcommand{\dfn}{Def.}
\newcommand{\fig}{Fig.}
\newcommand{\lem}{Lem.}
\newcommand{\thm}{Thm.}

\newcommand{\sect}{Sect.}

\newcommand{\bl}{[}
\newcommand{\br}{]}
\newcommand{\sbl}{\{}
\newcommand{\sbr}{\}}

\newcommand{\lika}{\mathbf{L}_{\Box\Diamond}(\axs)} %{\mathsf{LIK}(\axs)}

\newcommand{\ikal}{\mathsf{IK}(\axs)\mathsf{L}}

\newcommand{\nika}{\mathsf{NIK}(\axs)}

\newcommand{\inp}{\bullet}
\newcommand{\outp}{\circ}

\newcommand{\ns}{\Sigma}
\newcommand{\empseq}{\emptyset}

\newcommand{\botin}{(\bot^{\inp})}
\newcommand{\conin}{(\land^{\inp})}
\newcommand{\conout}{(\land^{\outp})}
\newcommand{\disin}{(\lor^{\inp})}
\newcommand{\disout}{(\lor^{\outp})}
\newcommand{\iimpin}{(\iimp^{\inp})}
\newcommand{\iimpout}{(\iimp^{\outp})}

\newcommand{\D}{(d)}

\newcommand{\boxout}{(\Box^{\outp})}
\newcommand{\diain}{(\dia^{\inp})}

\newcommand{\sar}{\vdash}

\newcommand{\h}{\mathsf{H}}
\newcommand{\axd}{\text{D}}

\newcommand{\axhsl}{\text{HSL}}
\newcommand{\axt}{\text{T}}
\newcommand{\axfour}{\text{4}}
\newcommand{\axfive}{\text{5}}
\newcommand{\axb}{\text{B}}
\newcommand{\gmp}{\phi(n,k)}
\newcommand{\hsl}{\phi(n,k)}
            
\newcommand{\ik}{\mathsf{IK}}

\newcommand{\ika}{\mathsf{IK}(\mathcal{A})}

\newcommand{\cate}{}
\newcommand{\concat}{\cdot}
\newcommand{\thuesys}{S}
\newcommand{\g}[1]{g(#1)}

\newcommand{\glang}{L_{\g{\axs}}}

\newcommand{\pto}{\longrightarrow}
\newcommand{\der}{\pto^{*}_{\g{\axs}}}
\newcommand{\fd}{\dia}
\newcommand{\bd}{\blacklozenge}
\newcommand{\ques}{\langle ? \rangle}

\newcommand{\prgr}[1]{PG(#1)}     
\newcommand{\prgrdom}{V}
\newcommand{\prgredges}{E}    
\newcommand{\empstr}{\varepsilon} 
\newcommand{\emppath}{\lambda} 
\newcommand{\ppath}{\pi}

\newcommand{\stra}{s}
\newcommand{\strb}{t}
\newcommand{\strc}{r}
\newcommand{\conv}[1]{{#1}^{-1}}
\newcommand{\albetstr}{\albet^{\ast}}

\newcommand{\osdr}{\pto_{\g{\axs}}}
\newcommand{\dr}{\pto_{\g{\axs}}^{*}}

\newcommand{\iimp}{\supset}

\newcommand{\ieq}{\equiv}
\newcommand{\inot}{{\sim}}

\newcommand{\dia}{\Diamond}

\newcommand{\albet}{\Upsigma}

\newcommand{\lang}{\mathcal{L}}

\newcommand{\prop}{\Upphi}

\newcommand{\axs}{\mathcal{A}}

\newcommand{\lab}{Lab}
\newcommand{\id}{(id)}
\newcommand{\botl}{(\bot_{l})}

\newcommand{\disr}{(\lor_{r})}
\newcommand{\conr}{(\land_{r})}

\newcommand{\boxr}{(\Box_{r})}
\newcommand{\boxl}{(\Box_{l})}
\newcommand{\diar}{(\dia_{r})}
\newcommand{\dial}{(\dia_{l})}

\newcommand{\iimpl}{(\iimp_{l})}
\newcommand{\iimpr}{(\iimp_{r})}
\newcommand{\disl}{(\lor_{l})}
\newcommand{\conl}{(\land_{l})}

\newcommand{\wk}{(w)}

\newcommand{\nec}{(n)}
\newcommand{\med}{(m)}
\newcommand{\ctr}{(c)}

\newcommand{\sa}{(S_{n,k})}

\newcommand{\prbox}{(p_{\Box})}
\newcommand{\prdia}{(p_{\dia})}

\newcommand{\rel}{\mathcal{R}}

%% file: introduction.tex
%%%NOTES:
%- Add Graph showing tansformations and also potentially soundness and compeleteness
%- New Results: (i) Formal grammars used to define propagation rules; gives more modular approach and answers open question of Lutz's paper to a high degree, (ii) Connects labelled and nested formalism, showing that latter can be extracted from former

%1.Introduction to uses of IMLs

Intuitionistic modal logics enable intuitionistic reasoning with the intensional operators $\dia$ and $\Box$. While a variety of different intuitionistic modal logics have been proposed~\cite{BiePai00,Fit48,PloSti86,Ser84,Sim94}, we focus on those defined in~\cite{PloSti86}, which extend the intuitionistic modal logic $\ik$ with Scott-Lemmon axioms~\cite{Lem77}. These logics were placed on a firm philosophical footing in~\cite{Sim94} due to their satisfaction of certain requirements that one might reasonably impose upon an intuitionistic version of modal logic. Although such logics are interesting in their own right, intuitionistic modal logics have proven useful in practical applications: having been applied in the verification of computer hardware~\cite{FaiMen95}, to facilitate reasoning about functional programs~\cite{Pit91}, and in defining programming languages~\cite{DavPfe01}.

%the natural inclination to blend modal and intuitionisitic reasoning, such logics have proven useful in practical applications, being applied in the verification of computer hardware~\cite{?}, to facilitate reasoning about functional programs~\cite{?}, and have proven useful in defining programming languages~\cite{?}.

%2. Importance of proof theory & itro to nested sequents

The development of intuitionistic modal logics naturally gave rise to an accompanying proof theory. Labelled natural deduction and sequent systems were provided for $\ik$ extended with geometric axioms in~\cite{Sim94}. In~\cite{GalSal10} and~\cite{GalSal15}, label-free natural deduction systems and tree-sequent calculi were respectively provided for extensions of $\ik$ with combinations of the reflexivity axiom ($\axt$), symmetry axiom ($\axb$), transitivity axiom ($\axfour$), and Euclidean axiom ($\axfive$). In~\cite{Str13}, nested sequent systems were proposed for all logics within the intuitionistic modal cube (i.e. logics axiomatized by extending $\ik$ with a subset of the axioms $\axt$, $\axb$, $\axfour$, $\axfive$, and the seriality axiom $\axd$). Such systems provide a suitable basis for developing automated reasoning and proof-search methods, having been used---in particular---to establish the decidability of logics within the intuitionistic modal cube~\cite{GalSal15,Sim94}.

With the exception of the systems introduced in~\cite{Sim94}, the drawback of the aforementioned proof systems is that they are rather limited, only being defined for a handful of logics. Indeed, in a recent paper on nested systems for intuitionistic modal logics~\cite{MarStr14}, the authors leave open the problem of defining rules within the nested sequent formalism that allow for the capture of logics \emph{outside} the intuitionistic modal cube. Accomplishing such a task would prove beneficial, since systems built within the nested formalism tend to be more economical (viz. they utilize simpler data structures) than those built within the labelled formalism, and have proven well-suited for the construction of analytic calculi~\cite{Bru09,Bul92,Kas94}, for writing decision algorithms~\cite{GalSal15,TiuIanGor12}, and for verifying interpolation~\cite{FitKuz15,LyoTiuGorClo20}.

%Regarding the labelled systems of~\cite{Sim94}, even though such systems were shown sound and complete relative to a truly extensive class of intuitionistic modal logics, the author was only able to leverage his proof systems to prove the decidability of extensions of $\ik$ with 

%3. Significance of refinment & propagation rules & how we go beyond state of the art (Lutz dissusses this in recent paper)

%4. Proof systems for IMLs and how ours differ/improve upon

In this paper, we answer the open problem of~\cite{MarStr14} to a large extent, and provide cut-free nested sequent systems for extensions of $\ik$ with what we call \emph{Horn-Scott-Lemmon axioms (HSLs)}, namely, axioms of the form $(\dia^{n} \Box A \iimp \Box^{k} A) \land (\dia^{k} A \iimp \Box^{n} \dia A)$. We obtain such systems through the recently developed \emph{structural refinement} methodology~\cite{Lyo21}, which consists of transforming a labelled sequent system into a nested system through the introduction of \emph{propagation rules} (cf.~\cite{CasCerGasHer97,Fit72}) and the elimination of structural rules, followed by a notational translation. The propagation rules operate by viewing labelled sequents (which encode binary labelled graphs) as automata, allowing for formulae to be propagated along a path in the underlying graph of a labelled sequent, so long as the path is encoded by a string derivable in a certain formal grammar. The refinement methodology grew out of works relating labelled systems to `more refined' or nested systems~\cite{CiaLyoRam18,GorRam12,LyoBer19,Pim18}. Also, the propagation rules we use are largely based upon the work of~\cite{GorPosTiu11,TiuIanGor12}, where such rules were used in the setting of display and nested calculi. These rules were then transported to the labelled setting to prove the decidability of agency logics~\cite{LyoBer19}, to establish translations between calculi within various proof-theoretic formalisms~\cite{CiaLyoRamTiu21}, and to provide a basis for the structural refinement methodology~\cite{Lyo21}.

This paper accomplishes the following: First, we show that structural refinement can be used to extract nested sequent systems from Simpson's labelled sequent systems~\cite{Sim94} with proofs in the latter formalism algorithmically translatable into proofs of the nested formalism.  %thus establishing a constructive relationship between the two formalisms. 
 Second, we provide sound and cut-free complete nested sequent systems for a considerable class of intuitionistic modal logics, and show that such systems admit the height-preserving admissibility (which we refer to as \emph{hp-admissibility}) of certain structural rules (e.g. forms of weakening and contraction). % that goes far beyond the state-of-the-art for such systems. 
 Third, we provide an answer to the open problem of~\cite{MarStr14} to a large degree, giving a straightforward procedure for transforming axioms (viz. HSLs) into propagation/logical rules.

We have organized this paper accordingly: In \sect~\ref{sec:log-prelims}, we define the intuitionistic modal logics considered, along with their axiomatizations and semantics. In \sect~\ref{sec:grammar-theory}, we introduce fundamental concepts in grammar theory that are needed for the definition of our propagation rules. We then introduce Simpson's labelled sequent calculi for intuitionistic modal logics in \sect~\ref{sec:labelled-systems}, and show how to structurally refine them in \sect~\ref{sec:refinement}. Last, in \sect~\ref{sec:nested-calculi}, we translate the refined labelled systems of the previous section into sound and cut-free complete nested sequent systems admitting the hp-admissibility of certain structural rules.

%%%For Journal Version
%- Can show hp-admissibility of structural rules, hp-invertibility of rules
%- Can show syntactic cut-elimination
%- Can give decidability for full class of logics (with counter-model extraction)

%% file: body-1.tex
%\resizebox{\columnwidth}{!}{
%%%NOTES
%1. Should mention that monotonic and nomial horn properties not only capture a notion of true formulae persisting into the future, but of false formula persisting into the past
%2. Should explain in a paragraph before properties are introduced what constitutes a model so that the first-order writing of the condition makes sense
%3. Should explain somewhere that \bot and \top are essentially propoisitonal variables unless we add falsum and verum. Also, will want to explain the conditions somewhere, and explain the intuition of the interpretation of each connective
%4. Include a table showing common properties that are among the ones we consider
%5. Include a table with the logics that are obtained by enforcing certain properties (sub-int logics of Restall and Corsi, Predicate logics of Ishigaki, Intuitionistic logic, Bi-intuitionistic logic, grammar logics, T, S4, S5, Kt, 'look-up-names-of-nominal logics', classical logic, 'free logics?', Dosens logic?, Vissers logic?)
%&. Explain that \chara and its converse are reserved for the bi-intuitionistic connectives

In this section, we introduce the language, semantics, and axiomatization for the intuitionistic modal logic $\ik$~\cite{PloSti86}.\footnote{See Simpson's 1994 PhD Thesis~\cite{Sim94} for a detailed introduction and discussion of $\ik$.} Moreover, we also introduce extensions of $\ik$ (referred to as \emph{intuitionistic modal logics} more generally) with the seriality axiom $\axd$ and axioms that we refer to as \emph{Horn-Scott-Lemmon Axioms ($\axhsl$s)}.

We define our intuitionistic modal language $\lang$ to be the set of formulae generated via the following BNF grammar:
$$
A ::= p \ | \ \bot \ | \ A \lor A \ | \ A \land A \ | \ A \iimp A \ | \ \dia A \ | \ \Box A
$$
where $p$ ranges over the set of propositional atoms $\prop := \{p, q, r, \ldots\}$. We use $A$, $B$, $C$, \etc (occasionally annotated) to range over formulae in $\lang$, and define $\inot A := A \iimp \bot$ and $A \ieq B := (A \iimp B) \land (B \iimp A)$. For $n \in \mathbb{N}$, we use $\dia^{n} A$ and $\Box^{n} A$ to represent the formula $A$ prefixed with a sequence of $n$ diamonds or boxes, respectively. We interpret such formulae on \emph{bi-relational models}~\cite{PloSti86,Sim94}: 

%Pg. 50 of Simpson
\begin{definition}[Bi-relational Model~\cite{PloSti86}]\label{def:bi-relational-model} We define a \emph{bi-relational model} to be a tuple $M := (W, \leq, R, V)$ such that:
\begin{itemize}

\item $W$ is a non-empty set of \emph{worlds} $w, u, v, \ldots$ (potentially annotated);

\item The \emph{intuitionistic relation} $\leq \ \subseteq W \times W$ is %a preorder, i.e. it is 
 reflexive and transitive;

\item The \emph{accessibility relation} $R \subseteq W \times W$ satisfies:

\begin{itemize}

\item[(F1)] For all $w, v, v' \in W$, if $w R v$ and $v \leq v'$, then there exists a $w' \in W$ such that $w \leq w'$ and $w' R v'$;

\item[(F2)] For all $w, w', v \in W$, if $w \leq w'$ and $w R v$, then there exists a $v' \in W$ such that $w' R v'$ and $v \leq v'$;

\end{itemize}

%\item For the set $\fworlds \subseteq W$ of \emph{fallible worlds}, if $w \in \fworlds$, and $wR_{\charx}u$ or $w \leq u$, then $u \in \fworlds$.

\item $V : W \to 2^{\prop}$ is a \emph{valuation function} satisfying the \emph{monotonicity condition}:  For each $w, u \in W$, if $w \leq u$, then $V(w) \subseteq V(u)$.

\end{itemize}

\end{definition}

%The (F1) and (F2) conditions can be visualized with the respective diagrams below. The dotted edges are those that are taken to exist given that the solid edges exist in a bi-relational model.

%\begin{center}
%\begin{tabular}{c @{\hskip 3em} c}
%\xymatrix{
% w'\ar@{.>}[rr]^{R}  & & v'   \\
% w\ar@{.>}[u]^{\leq}\ar[rr]_{R} &  & v\ar[u]_{\leq}
%}

%&

%\xymatrix{
% w'\ar@{.>}[rr]^{R} &  & v' \\
%w\ar[u]^{\leq}\ar[rr]_{R} & & v\ar@{.>}[u]_{\leq}
%}

%\end{tabular}
%\end{center}

Formulae from $\lang$ may then be interpreted over bi-relational models as specified by the semantic clauses below. %We note that in the intuitionistic setting the $\dia$ and $\Box$ modalities fail to be interdefinable (i.e. they fail to be dual to one another), that is, $\dia A \ieq \inot \Box \inot A$ is not valid according to the semantics for intuitionistic modal logics. This break down in the interdefinability of $\dia$ and $\Box$ is similar to the well-known fact that $\lor$ and $\land$, and $\exists$ and $\forall$, fail to be interdefinable in (first-order) intuitionistic logic.

%\cite{GabSheSkv09}
\begin{definition}[Semantic Clauses~\cite{PloSti86}]
\label{def:semantic-clauses} Let $M$ be a bi-relational model with $w \in W$ of $M$. The \emph{satisfaction relation} $M,w \Vdash A$ %between $w$ and a formula $A \in \lang$ 
 is defined recursively:

\begin{itemize}

\item $M,w \Vdash p$ \ifandonlyif $p \in V(w)$, for $p \in \prop$;

\item $M, w \not\Vdash \bot$;

\item $M,w \Vdash A \lor B$ \ifandonlyif $M,w \Vdash A$ or $M,w \Vdash B$;

\item $M,w \Vdash A \land B$ \ifandonlyif $M,w \Vdash A$ and $M,w \Vdash B$;

\item $M,w \Vdash A \iimp B$ \ifandonlyif for all $w' \in W$, if $w \leq w'$ and $M,w' \Vdash A$, then $M,w' \Vdash B$;

\item $M,w \Vdash \dia A$ \ifandonlyif there exists a $v \in W$ such that $w R v$ and $M,v \Vdash A$;

\item $M,w \Vdash \Box A$ \ifandonlyif for all $w', v' \in W$, if $w \leq w'$ %, then for all $v' \in W$, if 
 and $w' R v'$, then $M,v' \Vdash A$.

\end{itemize}
We say that a formula $A$ is \emph{globally true on $M$}, written $M \Vdash A$, \ifandonlyif $M,u \Vdash A$ for all worlds $u \in W$ of $M$, and we say that a formula $A$ is \emph{valid}, written $\Vdash A$, \ifandonlyif $A$ is globally true on all bi-relational models.
\end{definition}

As shown by Plotkin and Stirling in~\cite{PloSti86}, the validities of $\ik$ are axiomatizable:

%Pg. 52 Simpson
\begin{definition}[Axiomatization~\cite{PloSti86}] We define the axiomatization $\h\ik$ as:

\begin{multicols}{2}
\begin{itemize}

\item[A0] All theorems of propositional intuitionistic logic

\item[A1] $\Box (A \iimp B) \iimp (\Box A \iimp \Box B)$

\item[A2] $\Box (A \iimp B) \iimp (\dia A \iimp \dia B)$

\item[A3] $\inot \dia \bot$

\item[A4] $\dia (A \lor B) \iimp (\dia A \lor \dia B)$

\item[A5] $(\dia A \iimp \Box B) \iimp \Box (A \iimp B)$

\item[R0] \AxiomC{$A$}\AxiomC{$A \iimp B$}\RightLabel{(mp)}\BinaryInfC{$B$}\DisplayProof

\item[R1] \AxiomC{$A$}\RightLabel{(nec)}\UnaryInfC{$\Box A$}\DisplayProof

\end{itemize}
\end{multicols}
We define $\ik$ to be the smallest set of formulae closed under substitutions of the above axioms and applications of the inference rules, and define $A$ to be a \emph{theorem} of $\ik$ \ifandonlyif $A  \in \ik$.
\end{definition}

%Pg. 107 Simpson for Axioms and Frame Properties
We also consider %logics generated by 
 extensions of $\h\ik$ with sets $\axs$ of the following axioms:
$$
\axd: \Box A \iimp \dia A \quad \axhsl: (\dia^{n} \Box A \iimp \Box^{k} A) \land (\dia^{k} A \iimp \Box^{n} \dia A)
$$
The above left axiom is referred to as the \emph{seriality axiom $\axd$} and axioms of the form above right are referred to as \emph{Horn-Scott-Lemmon axioms ($\axhsl$s)}, which we use $\gmp$ to denote.\footnote{We note that the term \emph{Horn-Scott-Lemmon axiom} arises from the fact that such axioms form a proper subclass of the well-known \emph{Scott-Lemmon Axioms}~\cite{Lem77} and are associated with frame conditions that are expressible as Horn formulae~\cite[\sect~7.2]{Sim94}.} For the remainder of the paper, we use $\axs$ to denote an arbitrary set of the above axioms, that is:
$$
\axs \subseteq \{\axd\} \cup \{(\dia^{n} \Box A \iimp \Box^{k} A) \land (\dia^{k} A \iimp \Box^{n} \dia A) \ | \ n,k \in \mathbb{N} \}
$$
The set of $\axhsl$s includes well-known axioms such as:
%Below are given on Pg. 56 of Simpson
$$
\axt: (A \iimp \dia A) \land (\Box A \iimp A) \quad
\axfour: (\dia \dia A \iimp \dia A) \land (\Box A \iimp \Box \Box A)
$$
$$
\axb: (\dia \Box A \iimp A) \land (A \iimp \Box \dia A) \quad
\axfive: (\dia \Box A \iimp \Box A) \land (\dia A \iimp \Box \dia A)
$$
The work of Simpson~\cite{Sim94} establishes that any extension of $\h\ik$ with a set $\axs$ of axioms is sound and complete relative to a subclass of the bi-relational models. In particular, the extension of $\h\ik$ with a set $\axs$ of axioms is sound and complete relative to the set of bi-relational models satisfying the frame conditions related to the axioms of $\axs$, as specified in \fig~\ref{fig:axioms-related-conditions}.\footnote{We note that the axioms we consider do not \emph{characterize} the set of frames satisfying the frame properties related to the axioms as they do in the classical setting. For more details concerning this point, see~\cite[p.~56]{Sim94}, and for details concerning the proper characterization results of the above axioms, see~\cite{PloSti86}.} We define axiomatic extensions of $\h\ik$ along with their corresponding models below:

\begin{figure}[t]
\begin{center}
\bgroup
\def\arraystretch{1.1}
\begin{tabular}{| l | l |}
\hline
Axiom & Frame Condition\\
\hline
$\Box A \iimp \dia A$ & $\forall w \exists u (w R u)$\\
%\hline
$(\dia^{n} \Box A \iimp \Box^{k} A) \land (\dia^{k} A \iimp \Box^{n} \dia A)$ & $\forall w, u, v (w R^{n} u \land w R^{k} v \iimp u R v)$\\
\hline
\end{tabular}
\egroup
\end{center}

\caption{Axioms and their related frame conditions. We note that when $n=0$, the related frame condition is $\forall w, v (w R^{k} v \iimp w R v)$, when $k = 0$, the related frame condition is $\forall w, u (w R^{n} u \iimp u R w)$, and when $n = k = 0$, the related frame condition is $\forall w (w R w)$.}
\label{fig:axioms-related-conditions}
\end{figure}

\begin{definition}[Extensions, Bi-relational $\axs$-model, $\axs$-valid] The axiomatization $\h\ik(\axs)$ is defined to be $\h\ik$ extended with the axioms from $\axs$, and we define the \emph{logic} $\ika$ to be the smallest set of formulae closed under substitutions of the axioms of $\h\ika$ and applications of the inference rules. Also, a \emph{theorem} of $\ika$ is a formula $A$ such that $A \in \ika$. Moreover, we define a \emph{bi-relational $\axs$-model} to be a bi-relational model satisfying each frame condition related to an axiom $A \in \axs$ (as specified in \fig~\ref{fig:axioms-related-conditions}). Last, a formula $A$ is \emph{$\axs$-valid} \ifandonlyif it is globally true on all $\axs$-models.
\end{definition}

\begin{remark} We note that $\h\ik = \h\ik(\emptyset)$ and that a bi-relational $\emptyset$-model is a bi-relational model.
\end{remark}

%Pg. 107 Thm 6.2.1 of Simpson says "A is a theorem of IK + Axioms iff A is a theorem NDIK + Axioms", Pg. 78 has soundness Thm 4.8.1, and Pg. 89 should have compelteness expressed in Thm. 5.2.1 or Pg 86 with Thm. 5.1.1
%Thm 8.1.4 on Pg. 156 gives a type of completeness and soundness for extensions of the ND IK system
%Pg. 107 Simpson, but need to find explicit statement of this fact
\begin{theorem}[Soundness and Completeness~\cite{Sim94}]\label{thm:sound-complete-ikma}
A formula is a theorem of $\h\ika$ \ifandonlyif it is valid in all $\axs$-frames.
\end{theorem}

\begin{proof}
Follows from \thm~6.2.1 and \thm~8.1.4 of~\cite{Sim94}.
\qed
\end{proof}

%% file: body-2.tex
%As mentioned previously (and as is implied by the name), grammar logics connect concepts concerning formal grammars to logical concepts.  Therefore, due to the intimate connection between formal grammars and grammar logics, we will introduce additional formal language theoretic concepts below. Such concepts will allow us to establish a correspondence between formal grammars---in particular, specific types of \emph{Semi-Thue Systems} (\dfn~\ref{def:CFCST-kms})---with properties imposed on $\albet$-frames (\dfn~\ref{fig:frame-conditions-production-rules}), and to define new classes of grammar logics (\dfn~\ref{def:axiomatization-km}). Most definitions are taken from~\cite{DemNiv05}.

As will be seen later on (viz. in \sect~\ref{sec:refinement} and~\ref{sec:nested-calculi}), a central component to our refinement methodology---i.e. the extraction of nested calculi from labelled---is the use of inference rules whose applicability is determined on the basis of strings generated by a formal grammar. We therefore introduce grammar-theoretic notions that are essential to the functionality of such rules.

We let $\albet$ be our \emph{alphabet} consisting of the \emph{characters} $\fd$ and $\bd$, that is, $\albet := \{\fd,\bd\}$. The symbols $\fd$ and $\bd$ will be used to encode information about the accessibility relation $R$ of a bi-relational model in certain inference rules of our calculi. In particular, $\fd$ will be used to encode information about what is happening in the \emph{future} of the accessibility relation, and $\bd$ will be used to encode information about what is happening in the \emph{past} of the accessibility relation. We note that such symbols have been chosen due to their analogous meaning in the context of tense logics~\cite{GorPosTiu11,Kas94}. Also, following~\cite{GorPosTiu11}, we let $\ques \in \albet$ and $\conv{\ques} \in \albet \setminus \{\ques\}$, i.e. $\conv{\fd} := \bd$ and $\conv{\bd} := \fd$; we refer to $\fd$ and $\bd$ as \emph{converses} of one another. We may define \emph{strings} over our alphabet $\albet$ accordingly:

%~\cite{DemNiv05}
\begin{definition}[$\albetstr$] We let $\concat$ be the \emph{concatenation operation} with $\varepsilon$ the \emph{empty string}. We define the set $\albet^{*}$ of \emph{strings over $\albet$} to be the smallest set such that:
\begin{itemize}

\item $\albet \cup \{\varepsilon\} \subseteq \albet^{*}$

\item $\text{If } \stra \in \albet^{*} \text{ and } \ques \in \albet \text{, then } \stra \concat \ques \in \albet^{*}$

\end{itemize}
\end{definition}

For a set $\albetstr$ of strings, we use $\stra$, $\strb$, $\strc$, \etc (potentially annotated) to represent strings in $\albetstr$. Also, the empty string $\empstr$ is taken to be the identity element for the concatenation operation, i.e. $\stra \concat \empstr = \empstr \concat \stra = \stra$ for $\stra \in \albet^{*}$. Furthermore, we will not explicitly mention the concatenation operation in practice and let $\stra \cate \strb := \stra \concat \strb$, that is, we denote concatenation by simply gluing two strings together. Beyond concatenation, another useful operation to define on strings is the \emph{converse operation}, adapted from~\cite{TiuIanGor12}.

\begin{definition}[String Converse] We extend the converse operation to strings as follows:
\begin{itemize}

\item $\conv{\varepsilon} := \varepsilon$;

\item $\text{If } \stra = \ques_{1} \cdots \ques_{n} \text{, then } \conv{\stra} := \conv{\ques}_{n} \cdots \conv{\ques}_{1}$.

\end{itemize}
\end{definition}

%\begin{definition}[String Length] The \emph{length} of a string $\stra \in \albet^{*}$ is defined recursively: 
%\begin{itemize}

%\item $\lenstr{\stra} := 0$, if $\stra = \empstr$;

%\item $\lenstr{\stra} = \lenstr{\strb \cate \ques} := \lenstr{\strb} + 1$, if $\stra = \strb \cate \ques$ with $\strb \in \albet^{*}$ and $\ques \in \albet$.

%\end{itemize}
%\end{definition}

We let $\ques^{n}$ denote a string consisting of $n$ copies of $\ques$, which is $\empstr$ when $n = 0$. Making use of such notation, we can compactly define the notion of an \emph{$\axs$-grammar}, which encodes information contained in a set $\axs$ of axioms, and which will be employed in the definition of certain inference rules (see \sect~\ref{sec:refinement}).

\begin{definition}[$\axs$-grammar]\label{def:grammar} We define an \emph{$\axs$-grammar} to be a set $\g{\axs}$ such that:
\begin{center}
$(\dia \pto \bd^{n} \cate \dia^{k}), (\bd \pto \bd^{k} \cate \fd^{n}) \in \g{\axs}$ \ifandonlyif $(\dia^{n} \Box A \iimp \Box^{k} A) \land (\dia^{k} A \iimp \Box^{n} \dia A) \in \axs$.
\end{center}
We call rules of the form $\ques \pto \stra$ \emph{production rules}, where $\ques \in \albet$ and $\stra \in \albetstr$.
\end{definition}

%\begin{figure}[t]
%\begin{center}
%\bgroup
%\def\arraystretch{1.1}
%\begin{tabular}{| l | l | l |}
%\hline
%Name & Frame Property & Production Rules\\
%\hline
%Reflexivity & $\forall w wRw$ & $\{\ques \pto \empstr \ | \ \ques \in \albet\}$ \\
%Symmetry & $\forall w, u (wRu \iimp uRw)$ & $\{\ques \pto \conv{\ques} \ | \ \ques \in \albet\}$ \\
%Transitivity & $\forall w, v, u (wRv \land vRu \iimp wRu)$ & $\{\ques \pto \ques \cate \ques \ | \ \ques \in \albet\}$ \\
%Euclideanity & $\forall w, v, u (wRv \land wRu \iimp vRu)$ & $\{\ques \pto \conv{\ques} \cate \ques \ | \ \ques \in \albet\}$ \\
%\hline
%\end{tabular}
%\egroup
%\end{center}
%\caption{**double check these}
%\label{fig:frame-conditions-production-rules}
%\end{figure}

An $\axs$-grammar $\g{\axs}$ is a type of \emph{Semi-Thue system} (cf.~\cite{Pos47}), i.e. it is a string re-writing system. For example, assuming that $\ques \pto \stra \in \g{\axs}$, we may derive the string $\strb \cate \stra \cate \strc$ from $\strb \cate \ques \cate \strc$ in one-step by applying the mentioned production rule. As usual, through successive applications of production rules to a string $\stra \in \albetstr$, one obtains derivations of new strings, the collection of which, determines a language. We make such notions precise by means of the following definition:

%~\cite{DemNiv05}
\begin{definition}[Derivation, Language]\label{def:derivation-language} Let $\g{\axs}$ be an $\axs$-grammar. The \emph{one-step derivation relation} $\osdr$ holds between two strings $\stra$ and $\strb$ in $\albetstr$, written $\stra \osdr \strb$, \ifandonlyif there exist $\stra', \strb' \in \albetstr$ and $\ques \pto \strc \in \thuesys$ such that $\stra = \stra' \cate \ques \cate \strb'$ and $\strb = \stra' \cate \strc \cate \strb'$.  The \emph{derivation relation} $\dr$ is defined to be the reflexive and transitive closure of $\osdr$. For two strings $\stra, \strb \in \albetstr$, we refer to $\stra \dr \strb$ as a \emph{derivation of $\strb$ from $\stra$}, and define its \emph{length} to be equal to the minimal number of one-step derivations needed to derive $\strb$ from $\stra$ in $\g{\axs}$. Last, for a string $\stra \in \albet^{*}$, the \emph{language of $\stra$ relative to $\g{\axs}$} is defined to be the set $\glang(\stra) := \{\strb \ | \ \stra \dr \strb \}$.
\end{definition}

%~\cite{DemNiv05}
%\begin{definition}[Production Rule Satisfaction]\label{def:production-rule-sat-kms} We say that a $\albet$-model \emph{satisfies a production rule $a \pto \stra$} \ifandonlyif $R_{\stra} \subseteq R_{a}$, and say that a $\albet$-model \emph{satisfies a $\axsinp,\axsoutp)$-system $\thuesys$} \ifandonlyif it satisfies all production rules in $\thuesys$.
%\end{definition}

%~\cite{DemNiv05}
%\begin{definition}[$\thuesys$-validity]\label{def:S-validity-kms} Let $\thuesys$ be a $(\axsinp,\axsoutp)$-system. A formula $\phi \in \lang{\albet}$ is \emph{$\thuesys$-valid} \ifandonlyif for every $\albet$-model $M$, if $M$ satisfies $\thuesys$, then $M \Vdash \phi$.
%\end{definition}

%% file: body-3.tex
%We use the name $\lika$ to denote a labelled system as opposed to Simpson's name $\mathbf{L}_{\Box\Diamond}(\mathcal{T})$~\cite{Sim94}.

We introduce equivalent variants of Simpson's labelled sequent systems for intuitionistic modal logics~\cite{Sim94}, which are uniformly presented in \fig~\ref{fig:labelled-calculi}. We use the name $\lika$ to denote a labelled system as opposed to Simpson's name $\mathbf{L}_{\Box\Diamond}(\mathcal{T})$ since we define each system relative to a set $\axs$ of axioms (cf.~\cite{Sim94}). The sole difference between Simpson's original systems and the systems presented here is that we copy principal formulae into the premises of some rules. This minor change will facilitate our work in the subsequent section.

Simpson's systems make use of a denumerable set $\lab := \{w,u,v,\ldots\}$ of \emph{labels} (which we sometimes annotate), as well as two distinct types of formulae: \emph{labelled formulae}, which are of the form $w : A$ with $w \in \lab$ and $A \in \lang$, and \emph{relational atoms}, which are of the form $wRu$ for $w,u \in \lab$. We define a \emph{labelled sequent} to be a formula of the form $\rel, \Gamma \sar w : A$, where $\rel$ is a (potentially empty) multiset of relational atoms, and $\Gamma$ is a (potentially empty) multiset of labelled formulae. Also, we define a sequence of relational atoms $wR^{n}u := wRw_{1}, w_{1}Rw_{2}, \ldots, w_{n-1}Ru$, for $n \in \mathbb{N}$, and note that $wR^{0}u := (w = u)$.

%Pg. 126 Simpson has lab. seq. calculi
\begin{figure}[t]
\noindent\hrule

\begin{center}
\begin{tabular}{c c} % @{\hskip 1em} c}

\AxiomC{}
\RightLabel{$\id$}
\UnaryInfC{$\rel, w : p,\Gamma \sar w : p$}
\DisplayProof

&

\AxiomC{}
\RightLabel{$\botl$}
\UnaryInfC{$\rel, w : \bot,\Gamma \sar u : A$}
\DisplayProof
\end{tabular}
\end{center}

\begin{center}
%\resizebox{\columnwidth}{!}{
\begin{tabular}{c c}

\AxiomC{$\rel, \Gamma, w : A \sar u : C$}
\AxiomC{$\rel, \Gamma, w : B \sar u : C$}
\RightLabel{$\disl$}
\BinaryInfC{$\rel, \Gamma, w : A \vee B \sar u : C$}
\DisplayProof

&

\AxiomC{$\rel, \Gamma \sar w : A_{i}$}
\RightLabel{$\disr~i \in \{1,2\}$}
\UnaryInfC{$\rel, \Gamma \sar w : A_{1} \vee A_{2}$}
\DisplayProof
\end{tabular}
%}
\end{center}

\begin{center}
\begin{tabular}{c @{\hskip 1em} c @{\hskip 1em} c}

\AxiomC{$\rel, \Gamma, w :A, w :B \sar u : C$}
\RightLabel{$\conl$}
\UnaryInfC{$\rel \Gamma, w :A \wedge B \sar u : C$}
\DisplayProof

&

\AxiomC{$\rel, \Gamma \sar w : A$}
\AxiomC{$\rel, \Gamma \sar w : B$}
\RightLabel{$\conr$}
\BinaryInfC{$\rel, \Gamma \sar w : A \wedge B$}
\DisplayProof

\end{tabular}
\end{center}

\begin{center}
\begin{tabular}{c c}
\AxiomC{$\rel, \Gamma, w :A \iimp B \sar w :A$}
\AxiomC{$\rel, \Gamma, w : B \sar u : C$}
\RightLabel{$\iimpl$}
\BinaryInfC{$\rel, \Gamma, w :A \iimp B \sar u : C$}
\DisplayProof 

&

\AxiomC{$\rel, \Gamma, w : A \sar w : B$}
\RightLabel{$\iimpr$}
\UnaryInfC{$\rel, \Gamma \sar w : A \iimp B$}
\DisplayProof
\end{tabular}
\end{center}

\begin{center}
\begin{tabular}{c c}
\AxiomC{$\rel, w R u, \Gamma, u : A \sar v : B$}
\RightLabel{$\dial^{\dag}$}
\UnaryInfC{$\rel, \Gamma, w : \dia A \sar v : B$}
\DisplayProof

&

\AxiomC{$\rel, w R u, \Gamma \sar u : A$}
\RightLabel{$\diar$}
\UnaryInfC{$\rel, w R u, \Gamma \sar w : \dia A$}
\DisplayProof
\end{tabular}
\end{center}

\begin{center}
\begin{tabular}{c c}
\AxiomC{$\rel, w R u, \Gamma \sar u : A$}
\RightLabel{$\boxr^{\dag}$}
\UnaryInfC{$\rel, \Gamma \sar w : \Box A$}
\DisplayProof

&

\AxiomC{$\rel, w R u, \Gamma, w : \Box A, u : A \sar v : C$}
\RightLabel{$\boxl$}
\UnaryInfC{$\rel, w R u, \Gamma, w : \Box A \sar v : C$}
\DisplayProof
\end{tabular}
\end{center}

\begin{center}
\begin{tabular}{c c}
\AxiomC{$\rel, wRu, \Gamma \sar v : A$}
\RightLabel{$\D^{\dag}$}
\UnaryInfC{$\rel, \Gamma \sar v : A$}
\DisplayProof

&

\AxiomC{$\rel, w R^{n} u, w R^{k} v, u R v,  \Gamma \sar z : A$}
\RightLabel{$\sa$}
\UnaryInfC{$\rel, w R^{n} u, w R^{k} v, \Gamma \sar z : A$}
\DisplayProof
\end{tabular}
\end{center}

\hrule
\caption{The labelled calculi $\lika$. We have $\D$ as a rule in the calculus, if $\axd \in \axs$, and $\sa$ as a rule in the calculus, for each $\gmp \in \axs$. The side condition $\dag$ states that %the rule can be applied only if 
 $u$ must be an \emph{eigenvariable}, i.e. $u$ may not occur in the conclusion.} % of the inference.}
\label{fig:labelled-calculi}
\end{figure}%Pg. 126 Simpson has sequent calculus

We refer to the $\id$ and $\botl$ rules as \emph{initial rules}, to the $\D$ and $\sa$ rules as \emph{structural rules}, and to the remaining rules in \fig~\ref{fig:labelled-calculi} as \emph{logical rules}. Our use of the term \emph{structural rules} in reference to $\D$ and $\sa$ is consistent with the use of the term in the literature on proof systems for modal and related logics~\cite{Bru09,CiaLyoRamTiu21,GorPosTiu08,GorPosTiu11} and is based on the fact that such rules manipulate the underlying data structure of sequents as opposed to introducing more complex logical formulae. Also, we point out that the $\sa$ rules form a proper subclass of Simpson's $(S_{\chi})$ \emph{geometric structural rules}~(see~\cite[p.~126]{Sim94}) used to generate labelled sequent systems for $\ik$ extended with any number of \emph{geometric axioms}. When $n=0$ or $k=0$ in an $\axhsl$, i.e. when $\phi(0,k) \in \axs$, $\phi(n,0) \in \axs$, or $\phi(0,0) \in \axs$, the structural rules $(S_{0,k})$, $(S_{n,0})$, and $(S_{0,0})$ are defined accordingly:

%We note that although the term \emph{structural rule} has historically been used to categorize rules of a different form than $\D$ and $\sa$ (e.g. weakening and contraction~\cite{Tak13}), the usage in our context is consistent with the more contemporary usage in the literature~\cite{Bru09,CiaLyoRamTiu21,GorPosTiu08} and is based on the fact that such rules manipulate the underlying data structure of sequents as opposed to introducing more complex logical formulae. Also, we point out that the $\sa$ rules form a proper subclass of Simpson's $(S_{\chi})$ \emph{geometric structural rules}~(see~\cite[p.~126]{Sim94}) used to generate labelled sequent systems for $\ik$ extended with any number of \emph{geometric axioms}. When $n=0$ or $k=0$ in an $\axhsl$, i.e. when $\phi(0,k) \in \axs$, $\phi(n,0) \in \axs$, or $\phi(0,0) \in \axs$, the structural rules $(S_{0,k})$, $(S_{n,0})$, and $(S_{0,0})$ are defined accordingly:

\begin{center}
\begin{tabular}{c c}
\AxiomC{$\rel, w R^{k} v, w R v,  \Gamma \sar z : A$}
\RightLabel{$(S_{0,k})$}
\UnaryInfC{$\rel, w R^{k} v, \Gamma \sar z : A$}
\DisplayProof

&

\AxiomC{$\rel, w R^{n} u, u R w,  \Gamma \sar z : A$}
\RightLabel{$(S_{n,0})$}
\UnaryInfC{$\rel, w R^{n} u, \Gamma \sar z : A$}
\DisplayProof
\end{tabular}
\end{center}
\begin{center}
\AxiomC{$\rel, w R w,  \Gamma \sar z : A$}
\RightLabel{$(S_{0,0})$}
\UnaryInfC{$\rel, \Gamma \sar z : A$}
\DisplayProof
\end{center}

Let us now define the semantics for our labelled sequents, and then we state the soundness and completeness theorem for $\lika$.

%Interpretation, Satisfiability, Validity
\begin{definition}[Labelled Sequent Semantics]\label{def:sequent-semantics-kms} Let $M := (W, \leq, R, V)$ be a bi-relational $\axs$-model with $I : \ \lab \mapsto W$ an \emph{interpretation function} mapping labels to worlds.  We define the \emph{satisfaction} of relational atoms and labelled formulae:
\begin{itemize}

\item $M, I \models wRu$ \ifandonlyif $I(w) R I(u)$;

\item $M, I \models w : A$ \ifandonlyif $M, I(w) \Vdash A$.

\end{itemize}

A labelled sequent $\Lambda := \rel, \Gamma \sar v : B$ is \emph{satisfied} in $M$ with $I$, written $M,I \models \Lambda$, \ifandonlyif if $M, I \models wRu$ for all $wRu \in \rel$ and $M, I \models w : A$ for all $w : A \in \Gamma$, then $M, I \models v : B$. A labelled sequent $\Lambda$ is \emph{falsified} in $M$ with $I$ \ifandonlyif $M, I \not\models \Lambda$, that is, $\Lambda$ is not satisfied by $M$ with $I$.

Last, a labelled sequent $\Lambda$ is \emph{$\axs$-valid}, written $\models_{\axs} \Lambda$, \ifandonlyif it is satisfiable in every bi-relational $\axs$-model $M$ with every interpretation function $I$. We say that a labelled sequent $\Lambda$ is \emph{$\axs$-invalid} \ifandonlyif $\not\models_{\axs} \Lambda$, i.e. $\Lambda$ is not $\axs$-valid.
\end{definition}

%Pg. 127 Thm. 7.2.1 Simpsoin Completeness/Soundness of Sequent Calc wrt to ND
%Pg. 156 Th. 8.1.4 Simpson Completeness/Soundness of ND wrt Axiomatization
\begin{theorem}[$\lika$ Soundness and Completeness]\label{thm:lika-sound-complete}
$\rel, \Gamma \vdash w : A$ is derivable in $\lika$ \ifandonlyif $\rel, \Gamma \vdash w : A$ is $\axs$-valid.
\end{theorem}

\begin{proof}
Follows from \thm~7.2.1 and \thm~8.1.4 of~\cite{Sim94}.
\qed
\end{proof}

%Graph morphism (used for soundness) is defined on pg. 76 Simpson

%% file: body-4.tex
We show how to \emph{structurally refine} the labelled systems introduced in the previous section, that is, we implement a methodology introduced and applied in~\cite{CiaLyoRamTiu21,Lyo21,Lyo21thesis,LyoBer19} (referred to as \emph{structural refinement}, or \emph{refinement} more simply) for simplifying labelled systems and/or permitting the extraction of nested systems. The methodology consists of eliminating structural rules (viz. the $\sa$ rules in our setting) through the addition of \emph{propagation rules} (cf.~\cite{CasCerGasHer97,Fit72,Sim94}) to the labelled calculi, begetting systems that are translatable into nested systems.

The propagation rules we introduce are based on those of~\cite{CiaLyoRamTiu21,GorPosTiu11,Lyo21thesis,LyoBer19,TiuIanGor12}, and operate by viewing a labelled sequent as an automaton, allowing for the propagation of a formula (when applied bottom-up) from a label $w$ to a label $u$ given that a certain path of relational atoms exists between $w$ and $u$ (corresponding to a string generated by an $\axs$-grammar). We note that Simpson likewise introduced a variation of these rules, named $(\dia R)_{\mathcal{T}_{H}}$ and $(\Box L)_{\mathcal{T}_{H}}$ (see~\cite[p.~126]{Sim94}), by closing the relational atoms of a sequent under the frame conditions related to each $\axhsl$ $\hsl \in \axs$. We opt to use propagation rules based on formal grammars however because such rules permit the formulation of nested systems outside the class of HSL extensions of $\ik$, thus setting the stage for the construction of nested systems for even broader classes of logics in future work.\footnote{For instance, we could define our propagation rules relative to the formal grammar $\{\fd \pto \fd \bd\}$, which would give a calculus for a logic outside the class of HSL extensions of $\ik$.}

The definition of our propagation rules is built atop the notions introduced in the following two definitions: % (\dfn~\ref{def:propagation-graph} and~\ref{def:propagation-path}) and are based on similar definitions

%Propagation rules date back to the work of Fitting in~\cite{Fit72}, and obtain their name (coined in~\cite{CasCerGasHer97}) on the basis of their behavior, namely, when applied bottom-up such rules propagate formulae to different labels in a labelled sequent. Propagation rules hav 

\begin{definition}[Propagation Graph]\label{def:propagation-graph} The \emph{propagation graph} $\prgr{\rel}$ of a multiset of relational atoms $\rel$ is defined recursively on the structure of $\rel$:
%\begin{eqnarray*}
\begin{itemize}

\item $\prgr{\empseq} := (\emptyset, \emptyset)$;

\item $\prgr{wRu} := (\{w,u\}, \{(w,\fd,u),(u,\bd,w)\})$;

\item $\prgr{\rel_{1},\rel_{2}} := (V_{1}\cup V_{2}, E_{1} \cup E_{2}) \text{ where } PG_{x}(\rel_{i}) = (V_{i},E_{i})$.

\end{itemize}
%\end{eqnarray*}
We will often write $w \in \prgr{\rel}$ to mean $w \in \prgrdom$, and $(w,\ques,u) \in \prgr{\rel}$ to mean $(w,\ques,u) \in \prgredges$.
%We let $w \in \prgr{\rel}$ mean $w \in \prgrdom$ and $(w,\ques,u) \in \prgr{\rel}$ mean $(w,\ques,u) \in \prgredges$.
\end{definition}

\begin{definition}[Propagation Path]\label{def:propagation-path} %Let $\rel$ be a multiset of relational atoms. 
 We define a \emph{propagation path from $w_{1}$ to $w_{n}$ in $\prgr{\rel} := (V,E)$} to be a sequence of the following form:
$$
\ppath(w_{1},w_{n}) := w_{1}, \ques_{1}, w_{2}, \ques_{2}, \ldots , \ques_{n-1}, w_{n}
$$
such that $(w_{1}, \ques_{1}, w_{2}) , (w_{2}, \ques_{2}, w_{3}), \ldots, (w_{n-1}, \ques_{n-1}, w_{n}) \in \prgredges$. Given a propagation path of the above form, we define its \emph{converse} as shown below top and its \emph{string} as shown below bottom:
$$
\conv{\ppath}(w_{n},w_{1}) := w_{n}, \conv{\ques}_{n-1}, w_{n-1}, \conv{\ques}_{n-2}, \ldots, \conv{\ques}_{1}, w_{1}
$$
$$
\stra_{\ppath}(w_{1},w_{n}) := \ques_{1} \cate \ques_{2} \cate \cdots \cate \ques_{n-1}
$$
Last, we let $\emppath(w,w) := w$ represent an \emph{empty path} with the string of the empty path defined as $\stra_{\emppath}(w,w) := \empstr$.
\end{definition}

%\begin{definition}[Propagation Graph]\label{def:propagation-graph} Let $\Lambda = \rel, \Gamma \sar \Delta$ be a labelled sequent for grammar logics. We define the \emph{propagation graph} $\prgr{\Lambda} = (\prgrdom, \prgredges)$ to be the directed graph such that
%\begin{itemize}

%\item $\prgrdom := \lab(\Lambda)$;

%\item $\prgredges := \{(w,u,a), (u,w,\overline{a}) \ | \ R_{a}wu \in \rel \text{ or } R_{\conv{a}}uw \in \rel \}$.

%\end{itemize}
%We will often write $w \in \prgr{\Lambda}$ to mean $w \in \prgrdom$, and $(w,u,\chara) \in \prgr{\Lambda}$ to mean $(w,u,\chara) \in \prgredges$, for $\chara \in \albet$.
%\end{definition}

We are now in a position to define the operation of our propagation rules $\prdia$ and $\prbox$, which are displayed in \fig~\ref{fig:propagation-rules}. Each propagation rule $\prdia$ and $\prbox$ is applicable only if \emph{there exists a propagation path $\ppath(w,u)$ from $w$ to $u$ in the propagation graph $\prgr{\rel}$ such that the string $\stra_{\ppath}(w,u)$ is in the language $\glang(\fd)$.} We express this statement compactly by making use of its equivalent first-order representation:
$$
\exists \ppath(w,u) \in \prgr{\rel} (\stra_{\ppath}(w,u) \in \glang(\fd))
$$
We provide further intuition %concerning the operation of 
 regarding such rules by means of an example:

\begin{example}\label{ex:propagation-graph-path} Let $\rel := vRu, uRw$. We give a graphical depiction of $\prgr{\rel}$:

\begin{center}
\begin{minipage}[t]{.5\textwidth}
\xymatrix{
  v \ar@/^1pc/@{.>}[rr]|-{\dia} & &  u\ar@/^1pc/@{.>}[rr]|-{\fd} \ar@/^1pc/@{.>}[ll]|-{\bd} & &  w \ar@/^1pc/@{.>}[ll]|-{\bd}
}
\end{minipage}
%\begin{minipage}[t]{.05\textwidth}
%\
%\end{minipage}
\begin{minipage}[t]{.5\textwidth}
%\vspace{2.5em}
\begin{tabular}{c}
\AxiomC{ }
\noLine
\UnaryInfC{$\Lambda := vRu, uRw, w : \Box p, u : p \sar v : p \iimp q$}
\DisplayProof
\end{tabular}
\end{minipage}
\end{center}

Let $\axs := \{(\dia^{2} \Box A \iimp \Box^{1} A) \land (\dia^{1} A \iimp \Box^{2} \dia A)\}$, so that the corresponding $\axs$-grammar is $\g{\axs} = \{\fd \pto \bd \bd \fd, \bd \pto \bd \fd \fd \}$. Then, the path $\ppath(w,u) := w, \bd, u, \bd, v, \fd, u$ exists between $w$ and $u$. The first production rule of $\g{\axs}$ implies that $\stra_{\ppath}(w,u) = \bd \bd \fd \in \glang(\fd)$. Therefore, we are permitted to (top-down) apply the propagation rule $\prbox$ to $\Lambda$ to delete the labelled formula $u : p$, letting us derive $vRu, uRw, w : \Box p \sar v : p \iimp q$ .
\end{example}

\begin{figure}[t]
\noindent\hrule

\begin{center}
\AxiomC{$\rel, \Gamma \sar u : A$}
\RightLabel{$\prdia~\textit{ only if }~\exists \ppath(w,u) \in \prgr{\rel} (\stra_{\ppath}(w,u) \in \glang(\fd))$}
\UnaryInfC{$\rel, \Gamma \sar w : \dia A$}
\DisplayProof
\end{center}
\begin{center}
%\resizebox{\columnwidth}{!}{
\AxiomC{$\rel, \Gamma, w : \Box A, u : A \sar v : B$}
\RightLabel{$\prbox~\textit{ only if }~\exists \ppath(w,u) \in \prgr{\rel} (\stra_{\ppath}(w,u) \in \glang(\fd))$}
\UnaryInfC{$\rel, \Gamma, w : \Box A \sar v : B$}
\DisplayProof
\end{center}

\hrule
\caption{Propagation rules.}
\label{fig:propagation-rules}
\end{figure}

\begin{remark}\label{rem:diar-boxl-propagation-rule-instances}
The $\diar$ and $\boxl$ rules are instances of $\prdia$ and $\prbox$, respectively.
\end{remark}

\begin{definition}[Refined Labelled Calculus] We define the \emph{refined labelled calculus} $\ikal := \lika + \{\prdia,\prbox\} - \{\sa \ | \ \hsl \in \axs\}$.
%%%We define the \emph{refined labelled calculus} $\ikal$ to be the same as $\lika$, but with each $\sa$ rule omitted and the rules $\{\prdia,\prbox\}$ added. %We note that the functionality of the propagation rules $\prdia$ and $\prbox$ depends on the set $\axs$ of axioms that $\ikal$ takes as a parameter.
\end{definition}

We show that each calculus $\ikal$ is complete by means of a proof transformation procedure. That is, we show that through the elimination of structural rules we can transform a proof in $\lika$ into a proof in $\ikal$. We note that Simpson proved a similar result, showing that labelled derivations with structural rules are transformable into derivations with his propagation rules $(\dia R)_{\mathcal{T}_{H}}$ and $(\Box L)_{\mathcal{T}_{H}}$ (see~\cite[\sect~7.2]{Sim94}). In our context, the proof of structural rule eliminability requires more complex methods however due to the use of our new propagation rules that are parameterized with formal grammars. We first prove two crucial lemmata, and then show the elimination result.

\begin{lemma}\label{lem:permutation-1} Let $\rel_{1} := \rel, w R^{n} u, w R^{k} v, u R v$ and $\rel_{2} := \rel, w R^{n} u, w R^{k} v$. Suppose we are given a derivation in $\lika + \{\prdia, \prbox\}$ ending with:
\begin{center}
%\resizebox{\columnwidth}{!}{
\begin{tabular}{c}
\AxiomC{$\rel, w R^{n} u, w R^{k} v, u R v,  \Gamma \sar z : A$}
\RightLabel{$\prdia$}
\UnaryInfC{$\rel, w R^{n} u, w R^{k} v, u R v,  \Gamma \sar x : \dia A$}
\RightLabel{$\sa$}
\UnaryInfC{$\rel, w R^{n} u, w R^{k} v, \Gamma \sar x : \dia A$}
\DisplayProof
\end{tabular}
\end{center}
where the side condition $\exists \ppath(x,z) \in \prgr{\rel_{1}} (\stra_{\ppath}(x,z) \in \glang(\fd))$ holds due to $\prdia$. Then, $\exists \ppath'(x,z) \in \prgr{\rel_{2}} (\stra_{\ppath'}(x,z) \in \glang(\fd))$, that is to say, the $\sa$ rule is permutable with the $\prdia$ rule. %%%, as shown below:
%\begin{center}
%\begin{tabular}{c}
%\AxiomC{$\rel, w R^{n} u, w R^{k} v, u R v,  \Gamma \sar z : A$}
%\RightLabel{$\sa$}
%\UnaryInfC{$\rel, w R^{n} u, w R^{k} v,  \Gamma \sar z : A$}
%\RightLabel{$\prdia$}
%\UnaryInfC{$\rel, w R^{n} u, w R^{k} v, \Gamma \sar x : \dia A$}
%\DisplayProof
%\end{tabular}
%\end{center}
\end{lemma}

\begin{proof} We have two cases: either (i) the relational atom $u R v $ is not active in the $\prdia$ inference, or (ii) it is. Since (i) is easily resolved, we show (ii).

%%%(i) If the relational atom $u R v$ is not active in $\prdia$, then $u R v$ does not occur along the propagation path $\ppath(x,z)$. Therefore, there exists a propagation path $\ppath'(x,z) \in \prgr{\rel_{2}}$ (namely, $\ppath(x,z)$) such that $\stra_{\ppath'}(x,z) = \stra_{\ppath}(x,z) \in \glang(\fd)$, thus showing that $\sa$ may be applied first, and then $\prdia$ second, i.e. the two rules are permutable.

%%%(ii) 
 Let us suppose that the relational atom $u R v$ is active in $\prdia$, i.e. $u R v$ occurs along the propagation path $\ppath(x,z)$. To prove the claim, we need to show that $\exists \ppath'(x,z) \in \prgr{\rel_{2}} (\stra_{\ppath'}(x,z) \in \glang(\fd))$. Therefore, we construct such a propagation path by performing the following operations on $\ppath(x,z)$:
\begin{itemize}

\item replace each occurrence of $u, \dia, v$ in $\prgr{\rel_{1}}$ with

$$
u, \bd, u_{1}, \ldots, u_{n-1}, \bd, w, \fd, w_{1}, \ldots, w_{k-1}, \fd, v;
$$

\item replace each occurrence of $v, \bd, u$ in $\prgr{\rel_{1}}$ with 

$$
v, \bd, w_{k-1}, \ldots, w_{1}, \bd, w, \fd, u_{n-1}, \ldots, u_{1}, \fd, u.
$$
\end{itemize}
We let $\ppath'(x,z)$ denote the path obtained by performing the above operations on $\ppath(x,z)$, and note that first half of the first propagation path and the second half of the second propagation path correspond to the edges $(u, \bd, u_{1}), \ldots, (u_{n-1},\bd,w) \in \prgr{\rel_{1}}$ and $(w, \fd, u_{n-1}), \ldots, (u_{1},\fd,u) \in \prgr{\rel_{1}}$, respectively, obtained from the relational atoms $w R^{n} u \in \rel_{1}$, whereas the second half of the first propagation path and the first half of the second propagation path correspond to the edges $(w,\fd, w_{1}), \ldots, (w_{k-1},\fd, v) \in \prgr{\rel_{1}}$ and $(v,\bd, w_{k-1}), \ldots, (w_{1},\bd, w) \in \prgr{\rel_{1}}$, respectively, obtained from the edges $w R^{k} v \in \rel_{1}$ (by \dfn~\ref{def:propagation-graph}). Since the sole difference between $\prgr{\rel_{1}}$ and $\prgr{\rel_{2}}$ is that the former is guaranteed to contain the edges $(u,\fd,v)$ and $(v,\bd,u)$ obtained from $u R v$, while the latter is not, and since $\ppath'(x,z)$ omits the use of such edges (i.e. $u, \fd, v$ and $v,\bd, u$ do not occur in $\ppath'(x,z)$), we have that $\ppath'(x,z)$ is a propagation path in $\prgr{\rel_{2}}$.

To complete the proof, we need to additionally show that $\stra_{\ppath'}(x,z) \in \glang(\fd)$. By assumption, $\stra_{\ppath}(x,z) \in \glang(\fd)$, which implies that $\fd \der \stra_{\ppath}(x,z)$ by \dfn~\ref{def:derivation-language}. Since $\sa$ is a rule in $\lika$, it follows that $\fd \pto \bd^{n} \cate \fd^{k}$ and $\bd \pto \bd^{k} \cate \fd^{n} \in \g{\axs}$ by \dfn~\ref{def:grammar}. If we apply  $\fd \pto \bd^{n} \cate \fd^{k}$ to each occurrence of $\fd$ in $\stra_{\ppath}(x,z)$ corresponding to the edge $(u,\fd,v)$ (and relational atom $u R v$), and apply $\bd \pto \bd^{k} \cate \fd^{n}$ to each occurrence of $\bd$ in $\stra_{\ppath}(x,z)$ corresponding to the edge $(v,\bd,u)$ (and relational atom $u R v$), we obtain the string $\stra_{\ppath'}(x,z)$ and show that $\fd \dr \stra_{\ppath'}(x,z)$, i.e. $\stra_{\ppath'}(x,z) \in \glang(\fd)$.
\qed
\end{proof}

\begin{lemma}\label{lem:permutation-2} Let $\rel_{1} := \rel, w R^{n} u, w R^{k} v, u R v$ and $\rel_{2} := \rel, w R^{n} u, w R^{k} v$. Suppose we are given a derivation in $\lika + \{\prdia, \prbox\}$ ending with:
\begin{center}
%\resizebox{\columnwidth}{!}{
\begin{tabular}{c}
\AxiomC{$\rel, w R^{n} u, w R^{k} v, u R v, x : \Box A, y : A, \Gamma \sar z : C$}
\RightLabel{$\prbox$}
\UnaryInfC{$\rel, w R^{n} u, w R^{k} v, u R v, x : \Box A, \Gamma \sar z : C$}
\RightLabel{$\sa$}
\UnaryInfC{$\rel, w R^{n} u, w R^{k} v, x : \Box A, \Gamma \sar z : C$}
\DisplayProof
\end{tabular}
\end{center}
where the side condition $\exists \ppath(x,y) \in \prgr{\rel_{1}} (\stra_{\ppath}(x,y) \in \glang(\fd))$ holds due to $\prbox$. Then, $\exists \ppath'(x,y) \in \prgr{\rel_{2}} (\stra_{\ppath'}(x,y) \in \glang(\fd))$, that is to say, the $\sa$ rule is permutable with the $\prbox$ rule. %%%, as shown below:
%\begin{center}
%\AxiomC{$\rel, w R^{n} u, w R^{k} v, u R v, x : \Box A, y : A, \Gamma \sar z : C$}
%\RightLabel{$\sa$}
%\UnaryInfC{$\rel, w R^{n} u, w R^{k} v, x : \Box A, y : A, \Gamma \sar z : C$}
%\RightLabel{$\prbox$}
%\UnaryInfC{$\rel, w R^{n} u, w R^{k} v, x : \Box A, \Gamma \sar z : C$}
%\DisplayProof
%\end{center}
\end{lemma}

\begin{proof}
Similar to the proof of \lem~\ref{lem:permutation-1} above. 
\qed
\end{proof}

To improve the comprehensibility of the above lemmata, we provide an example of permuting an instance of the structural rule $\sa$ above an instance of a propagation rule.

\begin{example}\label{ex:permutation} Let $\axs := \{(\dia \Box A \iimp \Box A) \land (\dia A \iimp \Box \dia A)\}$ so that the $\axs$-grammar $\g{\axs} =  \{\fd \pto \bd \fd, \bd \pto \bd \fd\}$. In the top derivation below, we assume that $\prdia$ is applied due to the existence of the propagation path $\ppath(u,v) = u, \fd, v$ in $\prgr{wRu, wRv, uRv}$, where $\stra_{\ppath}(u,v)  = \fd \in \glang(\fd)$ by \dfn~\ref{def:derivation-language}. The propagation graph $\prgr{wRu, wRv, uRv}$ corresponding to the top sequent of the derivation shown below left is shown below right:

\begin{center}
\begin{minipage}{.45\textwidth}
\begin{center}
\begin{tabular}{c} %@{\hskip -.05em} c}
%\vspace*{1 em}
%\ \\
\AxiomC{}
\RightLabel{$\id$}
\UnaryInfC{$wRu, wRv, uRv, u : p \sar u : p$}
\RightLabel{$\prdia$}
\UnaryInfC{$wRu, wRv, uRv, u : p \sar v : \fd p$}
\RightLabel{$(S_{1,1})$}
\UnaryInfC{$wRu, wRv, u : p \sar v : \fd p$}
\DisplayProof
\end{tabular}
\end{center}
\end{minipage}
\begin{minipage}{.05\textwidth}
\ \\
\end{minipage}
\begin{minipage}{.45\textwidth}
\begin{center}
\begin{tabular}{c}
\xymatrix{
w \ar@/^-1pc/@{.>}[dd]|-{\fd} \ar@/^1pc/@{.>}[drr]|-{\fd} &	& \\
 & &  v \ar@/^1pc/@{.>}[ull]|-{\bd}\ar@/^1pc/@{.>}[dll]|-{\bd} \\
u \ar@/^-1pc/@{.>}[uu]|-{\bd}\ar@/^1pc/@{.>}[urr]|-{\fd} & &  
}
\end{tabular}
\end{center}
\end{minipage}
\end{center}

If we apply $\fd \pto \bd \fd \in \g{\axs}$ to $\stra_{\ppath}(u,v) = \fd$, then we obtain the string $\bd \fd$. Hence, $\fd \dr \bd \fd$, i.e. $\bd \fd \in \glang(\fd)$, meaning that a propagation path $\ppath'(u,v)$ ($= u, \bd, w, \fd, v$) exists in $\prgr{wRu,wRv}$ such that $\stra_{\ppath'}(u,v) = \bd \fd \in \glang(\fd)$. We may therefore apply $(a_{1,1})$ and then $\prdia$ as shown below left; the propagation graph $\prgr{wRu,wRv}$ is shown below right:

\begin{center}
\begin{minipage}{.45\textwidth}
\begin{center}
\begin{tabular}{c} %@{\hskip -.05em} c}
%\vspace*{1 em}
%\ \\
\AxiomC{}
\RightLabel{$\id$}
\UnaryInfC{$wRu, wRv, uRv, u : p \sar u : p$}
\RightLabel{$(S_{1,1})$}
\UnaryInfC{$wRu, wRv, u : p \sar u : p$}
\RightLabel{$\prdia$}
\UnaryInfC{$wRu, wRv, u : p \sar v : \fd p$}
\DisplayProof
\end{tabular}
\end{center}
\end{minipage}
\begin{minipage}{.05\textwidth}
\ \\
\end{minipage}
\begin{minipage}{.45\textwidth}
\begin{center}
\begin{tabular}{c}
%\begin{footnotesize}
\xymatrix{
w \ar@/^-1pc/@{.>}[dd]|-{\fd} \ar@/^1pc/@{.>}[drr]|-{\fd} &	& \\
 & &  v \ar@/^1pc/@{.>}[ull]|-{\bd} \\
u \ar@/^-1pc/@{.>}[uu]|-{\bd} & &  
}
%\end{footnotesize}
\end{tabular}
\end{center}
\end{minipage}
\end{center}

\end{example}

\begin{theorem}\label{thm:gtkms-to-kmsl-kms}
Every derivation in $\lika$ can be algorithmically transformed into a derivation in $\ikal$.
\end{theorem}

\begin{proof} We consider a derivation in $\lika$, which is a derivation in $\lika + \{\prdia, \prbox\}$. By Remark~\ref{rem:diar-boxl-propagation-rule-instances}, each instance of $\diar$ and $\boxl$ can be replaced by a $\prdia$ or $\prbox$ instance, respectively, meaning we may assume our derivation in $\lika + \{\prdia, \prbox\}$ is free of $\diar$ and $\boxl$ instances. We show that the derivation can be transformed into a derivation in $\ikal$ by induction on its height, that is, we consider a topmost occurrence of a structural rule $\sa$ and show that it can be eliminated. We obtain a derivation in $\ikal$ by successively eliminating topmost instances of $\sa$ rules.

\textit{Base case.} Observe that any application of $\sa$ to $\id$ or $\botl$ yields another instance of the rule.

\textit{Inductive step.} It is straightforward to verify that any instance of $\sa$ freely permutes above instances of all rules in $\lika + \{\prdia, \prbox\}$ with the exception of $\sa$, $\prdia$, and $\prbox$ (this follows from the fact that all other rules do not have active relational atoms in their conclusion). Since we are considering a topmost application of $\sa$, we need not consider the permutation of $\sa$ above another instance of $\sa$. The last two cases of permuting $\sa$ above $\prdia$ and $\prbox$ follow from \lem~\ref{lem:permutation-1} and~\ref{lem:permutation-2}, respectively. 
\qed
\end{proof}

\begin{theorem}[$\ikal$ Soundness and Completeness]\label{thm:ikal-sound-complete}
$\rel, \Gamma \vdash w : A$ is derivable in $\ikal$ \ifandonlyif $\rel, \Gamma \vdash w : A$ is $\axs$-valid.
\end{theorem}

\begin{proof}
The forward direction (soundness) is shown by induction on the height of the given derivation, and the backward direction (completeness) follows from \thm~\ref{thm:lika-sound-complete} and~\ref{thm:gtkms-to-kmsl-kms}. 
\qed
\end{proof}

%% file: body-5.tex
%NOTES:
% - Mention that we combine notation from Lutz 2013 and Alwen et al. 2012
%- might want to show that propagation rules subsume behavior of two input diamond rules and two output box rules (which are analogous to the rules with Alwen, Raj, and Linda's paper)
%- Might want to mention how having one formulae corresponds to having a single labelled formulae on the right in Simoson's formalism

In our setting, nested sequents are taken to be trees of multisets of formulae containing a unique formula that occupies a special status. We utilize the nested sequents of~\cite{Str13}, but note that the data structure underlying such sequents was originally used in~\cite{GalSal10}, and is similar to the nested sequents for classical modal logics employed in~\cite{Bru09}. Following~\cite{Str13}, we mark the special, unique formula with a white circle $\outp$ indicating that the formula is of \emph{output polarity}, and mark the other formulae with a black circle $\inp$ indicating that the formulae are of \emph{input polarity}. A nested sequent $\ns$ is defined via the following BNF grammars:
%\begin{definition}[Nested Sequent~\cite{Str13}]
$$
\ns ::= \Delta, \Pi \qquad \Delta ::= A_{1}^{\inp}, \ldots, A_{n}^{\inp}, \bl \Delta_{1} \br, \ldots, \bl \Delta_{k} \br \qquad \Pi ::= A^{\outp} \ | \ \bl \ns \br
$$
%\end{definition}

We assume that the comma operator associates and commutes, implying that such sequents are truly trees of multisets of formulae, and we let the \emph{empty sequent} be the empty multiset $\empseq$. We refer to a sequent in the shape of $\Delta$ (which contains only input formulae) as an \emph{LHS-sequent}, a sequent in the shape of $\Pi$ %%%(which is either an output formula or a sequent $\bl \ns \br$) 
 as an \emph{RHS-sequent}, and a sequent $\ns$ as a \emph{full sequent}. We use both $\ns$ and $\Delta$ to denote LHS- and full sequents with the context differentiating the usage.

As for classical modal logics (e.g.~\cite{Bru09,GorPosTiu11}), we define a \emph{context} $\ns\{ \ \}\cdots\{ \ \}$ to be a nested sequent with some number of holes $\{ \ \}$ in the place of formulae. This gives rise to two types of contexts: \emph{input contexts}, which require holes to be filled with LHS-sequents to obtain a full sequent, and \emph{output contexts}, which require a single hole to be filled with an RHS-sequent and the remaining holes to be filled with LHS-sequents to obtain a full sequent. We also define the \emph{output pruning} of an input context $\ns\{ \ \}\cdots\{ \ \}$ or full sequent $\ns$, denoted $\ns^{\downarrow}\{ \ \}\cdots\{ \ \}$ and $\ns^{\downarrow}$ respectively, to be the same context or sequent with the unique output formula deleted. We note that all of the above terminology is due to~\cite{Str13}.

\begin{example} Let $\ns_{1}\{\} := p^{\inp}, \bl \dia q^{\inp}, \{ \ \} \br$, $\ns_{2}\{\}:= p^{\inp}, \bl \dia q^{\outp}, \{ \ \} \br$, $\Delta_{1} := \bot^{\inp}, \bl q \iimp r^{\outp} \br$, and $\Delta_{2} := \bot^{\inp}, \bl q \iimp r^{\inp} \br$. Observe that neither $\ns_{1}\{\Delta_{2}\}$ nor $\ns_{2}\{\Delta_{1}\}$ are full sequents since the former has no output formula and the latter has two output formulae. Conversely, both $\ns_{1}\{\Delta_{1}\}$ and $\ns_{2}\{\Delta_{2}\}$ are full sequents.
\end{example}

%For more information on nested sequents of the above form, consult~\cite{Sta13}.

Our nested sequent systems are presented in \fig~\ref{fig:nested-calculus} and are generalizations of those for the the logics of the intuitionistic modal cube given in~\cite{Str13}. For example, a nested sequent system for the intuitionistic modal logic $\ik + \{(\dia^{0} \Box A \iimp \Box^{3} A) \land (\dia^{3} A \iimp \Box^{0} \dia A)\}$ incorporating the 3-to-1 transitivity axiom, which falls outside the intuitionistic modal cube, is obtained by employing the $\axs$-grammar $\g{\axs} = \{\fd \pto \fd \fd \fd, \bd \pto \bd \bd \bd \}$ in the propagation rules $\prdia$ and $\prbox$. As in the previous section, our propagation rules $\prdia$ and $\prbox$ rely on auxiliary notions (e.g. propagation graphs and paths), which we define for nested sequents.

\begin{definition}[Propagation Graph/Path]\label{def:propagation-graph} Let $w$ be the label assigned to the root of the nested sequent $\ns$. We define the \emph{propagation graph} $PG(\ns) := PG_{w}(\ns)$ of a nested sequent $\ns$ recursively on the structure of the nested sequent.
\begin{itemize}

\item $PG_{u}(\empseq) := (\emptyset, \emptyset, \emptyset)$;

\item $PG_{u}(A) := (\emptyset, \emptyset, \{(u,A)\}) \text{ with } A \in \{A^{\inp}, A^{\outp}\}$;

\item $PG_{u}(\Delta_{1},\Delta_{2}) := (V_{1}\cup V_{2}, E_{1} \cup E_{2}, L_{1} \cup L_{2}) \text{ where } PG_{u}(\Delta_{i}) = (V_{i},E_{i},L_{i})$;

\item $PG_{u}(\bl \ns \br) := (V \cup \{u\}, E \cup \{(u,\fd,v),(v,\bd,u)\},L) \text{ where } PG_{v}(\ns) = (V,E,L)$ $\text{ and $v$ is fresh}$.

\end{itemize}
We will often write $u \in \prgr{\ns}$ to mean $u \in \prgrdom$, and $(u,\ques,v) \in \prgr{\ns}$ to mean $(u,\ques,v) \in \prgredges$. Also, we define \emph{propagation paths}, \emph{converses} of propagation paths, and the \emph{string} of a propagation path as in \dfn~\ref{def:propagation-path}.
\end{definition}

For input or output formulae $A$ and $B$, we use the notation $\ns\{A\}_{w}\{B\}_{u}$ to mean that $(w,A), (u,B) \in L$ in $\prgr{\ns}$. For example, if $\ns := p \iimp q^{\outp}, \bl p^{\inp}, \bl \Box p^{\inp} \br \br$ with $\prgr{\ns} := (V,E,L)$ and $(v,p \iimp q^{\outp}),(u, p^{\inp}), (w,\Box p^{\inp}) \in L$, then both $\ns\{p \iimp q^{\outp}\}_{v}\{\Box p^{\inp}\}_{w}$ and $\ns\{p^{\inp}\}_{u}\{p \iimp q^{\outp}\}_{v}$ are valid representations of $\ns$ in our notation.
%\begin{center}
%\xymatrix{
% & &  \overset{\boxed{p}}{w}\ar[drr]\ar[dll]\ar@/^-1pc/@{.>}[dll]|-{\fd}\ar@/^1pc/@{.>}[drr]|-{\fd} & & \\
%  \overset{\boxed{\bot}}{u} \ar@/^-1pc/@{.>}[urr]|-{\bd} & &  & & \overset{\boxed{\dia q}}{v} \ar@/^1pc/@{.>}[ull]|-{\bd}
%}
%\end{center}

\begin{figure}[t]
\noindent\hrule

\begin{center}
\begin{tabular}{c c c c} % @{\hskip 1em} c}

\AxiomC{}
\RightLabel{$\botin$}
\UnaryInfC{$\ns \sbl \bot^{\inp} \sbr$}
\DisplayProof

&

\AxiomC{}
\RightLabel{$\id$}
\UnaryInfC{$\ns \sbl p^{\inp}, p^{\outp} \sbr$}
\DisplayProof

&

\AxiomC{$\ns \sbl A^{\inp}, B^{\inp} \sbr$}
\RightLabel{$\conin$}
\UnaryInfC{$\ns \sbl A \land B^{\inp} \sbr$}
\DisplayProof

&

\AxiomC{$\ns \sbl A^{\outp} \sbr$}
\AxiomC{$\ns \sbl B^{\outp} \sbr$}
\RightLabel{$\conout$}
\BinaryInfC{$\ns \sbl A \land B^{\outp} \sbr$}
\DisplayProof
\end{tabular}
\end{center}

\begin{center}
\begin{tabular}{c c c c}
%\AxiomC{$\ns \sbl A^{\inp}, \inot A^{\outp} \sbr$}
%\RightLabel{$\negout$}
%\UnaryInfC{$\ns \sbl \inot A^{\outp} \sbr$}
%\DisplayProof

%&

%\AxiomC{$\ns^{\da} \sbl \inot A^{\inp}, A^{\outp} \sbr$}
%\RightLabel{$\negin$}
%\UnaryInfC{$\ns \sbl \inot A^{\inp} \sbr$}
%\DisplayProof

\AxiomC{$\ns \sbl A^{\inp} \sbr$}
\AxiomC{$\ns \sbl B^{\inp} \sbr$}
\RightLabel{$\disin$}
\BinaryInfC{$\ns \sbl A \lor B^{\inp} \sbr$}
\DisplayProof

&

\AxiomC{$\ns \sbl A^{\inp}, B^{\outp} \sbr$}
\RightLabel{$\iimpout$}
\UnaryInfC{$\ns \sbl A \iimp B^{\outp} \sbr$}
\DisplayProof

&

\AxiomC{$\ns \sbl A_{i}^{\outp} \sbr$}
\RightLabel{$\disout~i \in \{1,2\}$}
\UnaryInfC{$\ns \sbl A_{1} \lor A_{2}^{\outp} \sbr$}
\DisplayProof
\end{tabular}
\end{center}

\begin{center}
\begin{tabular}{c c c c}
\AxiomC{$\ns^{\downarrow} \sbl A \iimp B^{\inp}, A^{\outp} \sbr$}
\AxiomC{$\ns \sbl B^{\inp} \sbr$}
\RightLabel{$\iimpin$}
\BinaryInfC{$\ns \sbl A \iimp B^{\inp} \sbr$}
\DisplayProof

&

\AxiomC{$\ns \sbl \bl A^{\outp} \br \sbr$}
\RightLabel{$\boxout$}
\UnaryInfC{$\ns \sbl \Box A^{\outp} \sbr$}
\DisplayProof

&

\AxiomC{$\ns \sbl \bl A^{\inp} \br \sbr$}
\RightLabel{$\diain$}
\UnaryInfC{$\ns \sbl \dia A^{\inp} \sbr$}
\DisplayProof

&

\AxiomC{$\ns \sbl \bl \empseq \br \sbr$}
\RightLabel{$\D$}
\UnaryInfC{$\ns \sbl \empseq \sbr$}
\DisplayProof
\end{tabular}
\end{center}

\begin{center}
\AxiomC{$\ns \sbl \Delta_{1} \sbr_{w} \sbl A^{\outp}, \Delta_{2} \sbr_{u}$}
\RightLabel{$\prdia~\textit{ only if }~\exists \ppath(w,u) \in \prgr{\ns} (\stra_{\ppath}(w,u) \in \glang(\fd))$}
\UnaryInfC{$\ns \sbl \dia A^{\outp}, \Delta_{1} \sbr_{w} \sbl \Delta_{2} \sbr_{u}$}
\DisplayProof
\end{center}
\begin{center}
%\resizebox{\columnwidth}{!}{
\AxiomC{$\ns \sbl \Box A^{\inp}, \Delta_{1} \sbr_{w} \sbl A^{\inp}, \Delta_{2} \sbr_{u}$}
\RightLabel{$\prbox~\textit{ only if }~\exists \ppath(w,u) \in \prgr{\ns} (\stra_{\ppath}(w,u) \in \glang(\fd))$}
\UnaryInfC{$\ns \sbl \Box A^{\inp}, \Delta_{1} \sbr_{w} \sbl \Delta_{2} \sbr_{u}$}
\DisplayProof
\end{center}

\hrule
\caption{The nested sequent calculi $\nika$. The $\D$ rule occurs in a calculus \ifandonlyif $\axd \in \axs$.}
\label{fig:nested-calculus}
\end{figure}

We now prove that proofs can be translated between our refined labelled and nested systems. In order to prove this fact, we make use of the following definitions, which are based on the work of~~\cite{GorRam12,Lyo21}.

\begin{definition}[Labelled Tree Sequent/Derivation]
We define a \emph{labelled tree sequent} to be a labelled sequent $\Lambda := \rel, \Gamma \sar w : A$ such that $\rel$ forms a tree and all labels in $\Gamma, w : A$ occur in $\rel$. We define a \emph{labelled tree derivation} to be a proof containing only labelled tree sequents. We say that a labelled tree derivation has the \emph{fixed root property} \ifandonlyif every labelled sequent in the derivation has the same root.
\end{definition}

We now define our translation functions which transform a \emph{full} nested sequent into a labelled tree sequent, and vice-versa. Our translations additionally depend on \emph{sequent compositions} and \emph{labelled restrictions}. If $\Lambda_{1} := \rel_{1}, \Gamma_{1} \sar \Gamma_{1}'$ and $\Lambda_{2} := \rel_{2}, \Gamma_{2} \sar \Gamma_{2}'$, then we define its sequent composition $\Lambda_{1} \seqcomp \Lambda_{2} := \rel_{1}, \rel_{2},\Gamma_{1}, \Gamma_{2} \sar \Gamma_{1}', \Gamma_{2}'$. Given that $\Gamma$ is a multiset of labelled formulae, we define the labelled restriction $\Gamma \restriction w := \{A \ | \ w : A \in \Gamma\}$, and we note that if $w$ is not a label in $\Gamma$, then $\Gamma \restriction w := \emptyset$. %Moreover, for $\ast \in \{\inp,\outp\}$, we recursively define $(\emptyset)^{\ast} := \emptyset$ and $(A_{1}, \ldots, A_{n-1}, A_{n})^{\ast} := (A_{1}, \ldots,A_{n-1})^{\ast}, A_{n}^{\ast}$. 
 Moreover, for a multiset $A_{1}, \ldots,A_{n}$ of formulae, we define $(A_{1}, \ldots,A_{n})^{\ast} := A_{1}^{\ast}, \ldots,A_{n}^{\ast}$ and $(\emptyset)^{\ast} := \emptyset$, where $\ast \in \{\inp,\outp\}$.

\begin{definition}[Translation $\ltr$] We define $\ltr_{w}(\ns) := \rel, \Gamma \sar u : A$ as follows:

\begin{center}
\begin{minipage}{.45\textwidth}
\begin{itemize}

\item $\ltr_{v}(\empseq) := \emptyset \sar \emptyset$

\item $\ltr_{v}(A^{\inp}) := v : A \sar \empseq$

\item $\ltr_{v}(A^{\outp}) := \empseq \sar v : A$

\end{itemize}
\end{minipage}
\begin{minipage}{.5\textwidth}
\begin{itemize}

\item $\ltr_{v}(\Delta_{1},\Delta_{2}) := \ltr_{v}(\Delta_{1}) \seqcomp \ltr_{v}(\Delta_{2})$

\item $\ltr_{v}(\bl \ns \br) := (vRu \sar \empseq) \seqcomp \ltr_{u}(\ns)$ $\text{ with }$ $\text{$u$ fresh}$

\end{itemize}
\end{minipage}
\end{center}

We note that since $\ns$ is a full sequent, the obtained labelled sequent will contain a single labelled formula in its consequent.
\end{definition}

\begin{example} We let $\ns := p \iimp q^{\outp}, \bl p^{\inp}, \bl \Box p^{\inp} \br \br$ and show the output labelled sequent under the translation $\ltr$.
$$
\ltr_{w}(\Sigma) = wRv,vRu, v : p, u : \Box p \sar w : p \iimp q
$$
\end{example}

\begin{definition}[Translation $\ntr$] Let $\Lambda := \rel, \Gamma \sar w : A$ be a labelled tree sequent with root $u$. We define $\Lambda_{1} \subseteq \Lambda$ \ifandonlyif there exists a labelled tree sequent $\Lambda_{2}$ such that $\Lambda = \Lambda_{1} \seqcomp \Lambda_{2}$. Let us further define $\Lambda_{u}$ to be the unique labelled tree sequent rooted at the label $u$ such that $\Lambda_{u} \subseteq \Lambda$. We define $\ntr(\Lambda) := \ntr_{u}(\Lambda)$ recursively on the tree structure of $\Lambda$:
\[
  \ntr_{v}(\Lambda) :=
  \begin{cases}
  (\Gamma \restriction v)^{\inp}, (w : A \restriction v)^{\outp} & \text{if $\rel = \empseq$}; \\
  (\Gamma \restriction v)^{\inp}, (w : A \restriction v)^{\outp}, \bl \ntr_{z_{1}}(\Lambda_{z_{1}}) \br, \ldots, \bl \ntr_{z_{n}}(\Lambda_{z_{n}}) \br & \text{otherwise}. 
  \end{cases}
\]
In the second case above, we assume that $vRz_{1}, \ldots vRz_{n}$ are all of the relational atoms occurring in the input sequent which have the form $vRx$.
\end{definition}

\begin{example} We let $\Lambda := wRv,vRu, v : p, u : \Box p \sar w : p \iimp q$ and show the output nested sequent under the translation $\ntr$.
$$
\ntr(\Lambda) = \ntr_{w}(\Lambda) = p \iimp q^{\outp}, \bl p^{\inp}, \bl \Box p^{\inp} \br \br
$$
\end{example}

\begin{lemma}\label{lem:labelled-tree-derivations}
Every proof in $\ikal$ of a labelled tree sequent is a labelled tree proof with the fixed root property.
\end{lemma}

\begin{proof}
The lemma follows from the observation that %the only rules that add edges in $\ikal$ are $\dial$ and $\boxr$, which 
 if any rule of $\ikal$ is applied bottom-up to a labelled tree sequent, then each premise is a labelled tree sequent with the same root. 
\qed
\end{proof}

\begin{theorem}\label{thm:nested-labelled-equiv}
Every proof of a labelled tree sequent in $\ikal$ is transformable into a proof in $\nika$, and vice-versa.
\end{theorem}

\begin{proof}
Follows from \lem~\ref{lem:labelled-tree-derivations}, and the fact that the rules of $\ikal$ and $\nika$ are translations of one another under the $\ntr$ and $\ltr$ functions. 
\qed
\end{proof}

\begin{theorem}[$\nika$ Soundness and Completeness]
A formula $A$ is derivable in $\nika$ \ifandonlyif $A$ is $\axs$-valid.
\end{theorem}

\begin{proof}
Follows from \thm~\ref{thm:ikal-sound-complete} and~\ref{thm:nested-labelled-equiv}. 
\qed
\end{proof}

\begin{figure}[t]
\noindent\hrule

\begin{center}
\begin{tabular}{c c c c}
\AxiomC{$\ns$}
\RightLabel{$\nec$}
\UnaryInfC{$\bl \ns \br$}
\DisplayProof

&

\AxiomC{$\ns \sbl \empseq \sbr$}
\RightLabel{$\wk$}
\UnaryInfC{$\ns \sbl \Delta \sbr$}
\DisplayProof

&

\AxiomC{$\ns \sbl A^{\inp}, A^{\inp} \sbr$}
\RightLabel{$\ctr$}
\UnaryInfC{$\ns \sbl A^{\inp} \sbr$}
\DisplayProof

&

\AxiomC{$\ns \sbl \bl \Delta_{1}\br,  \bl \Delta_{2} \br \sbr$}
\RightLabel{$\med$}
\UnaryInfC{$\ns \sbl \bl \Delta_{1}, \Delta_{2} \br \sbr$}
\DisplayProof
\end{tabular}
\end{center}

\hrule
\caption{Height-preserving (hp-)admissible structural rules.} %(cf.~\ref{Str13}).}
\label{fig:structural-rules}
\end{figure}

\begin{theorem}
The rules $\nec$, $\wk$, $\ctr$, and $\med$ are hp-admissible in $\nika$.
\end{theorem}

\begin{proof} The height-preserving (hp-)admissibility of each rule (displayed in \fig~\ref{fig:structural-rules}) is shown by induction on the height of the given derivation. For the $\med$ rule, we note that propagation paths are preserved from premise to conclusion (cf.~\cite[\fig~12]{GorPosTiu11}), showing that the rule can be permuted above $\prdia$ and $\prbox$. 
\qed
\end{proof}

%% file: conclusion.tex
In this paper, we employed the structural refinement methodology to extract nested sequent systems for a broad class of intuitionistic modal logics. The attainment of such systems answers the open problem of~\cite{MarStr14} to a large extent by showing how to transform axioms (namely, HSLs) into propagation/logical rules as well as how to obtain nested sequent systems for logics outside the intuitionistic modal cube. We aim to write proof-search algorithms in future work based on our nested systems which utilize saturation conditions and loop-checking (cf.~\cite{Fit83,Lyo21thesis,TiuIanGor12}) to provide decision procedures for logics within the class considered. Our primary concern will be to establish the decidability of transitive extensions of $\ik$, which has remained a longstanding open problem~\cite{GirStr20,Sim94}.